\documentclass[reqno]{amsart} 
\usepackage{amsmath,amssymb}

\theoremstyle{plain}
\newtheorem{theorem}{Theorem}[section]

\newtheorem{lemma}[theorem]{Lemma}
\newtheorem{corollary}[theorem]{Corollary}
\theoremstyle{remark}
\newtheorem{remark}[theorem]{Remark}
\numberwithin{equation}{section}
\numberwithin{theorem}{section}

\newcommand{\mc}[1]{{\mathcal #1}}
\newcommand{\bs}[1]{{\mathbf #1}}
\newcommand{\bb}[1]{{\mathbb #1}}
\newcommand{\rme}{\mathrm{e}}
\newcommand{\rmd}{\mathrm{d}}

\title[Infinite Cucker-Smale-type dynamics with repulsions]{Cucker-Smale type dynamics of infinitely many individuals with repulsive forces}

\author[P.\ Butt\`a]{Paolo Butt\`a}
\address{Paolo Butt\`a\hfill\break \indent
Dipartimento di Matematica, 
Sapienza Universit\`a di Roma,
\hfill\break \indent
P.le Aldo Moro 5, 00185 Roma, Italy}
\email{butta@mat.uniroma1.it}
\author[C.\ Marchioro]{Carlo Marchioro}
\address{Carlo Marchioro \hfill\break \indent
Dipartimento di Matematica, 
Sapienza Universit\`a di Roma,
\hfill\break \indent
P.le Aldo Moro 5, 00185 Roma, Italy \hfill\break \indent and \hfill\break \indent  
International Research Center M\&MOCS, Universit\`a di L'Aquila \hfill\break \indent Palazzo Caetani, 04012 Cisterna di Latina (LT), Italy}
\email{marchior@mat.uniroma1.it}

\begin{document}
	
\begin{abstract}
We study the existence and uniqueness of the time evolution of a system of infinitely many individuals, moving in a tunnel and subjected to a Cucker-Smale type alignment dynamics with compactly supported communication kernels and to short-range repulsive interactions to avoid collisions.
\end{abstract}
	
\keywords{Infinite dynamics, Cucker-Smale, Individual-based models, Flocking}

\subjclass[2010]{
Primary: 82C22. 
Secondary: 92D25,
37N25
}
	
\maketitle
\thispagestyle{empty}

\section{Introduction}
\label{sec:1}

The Cucker-Smale (CS) model \cite{CS1,CS2} is a deterministic system that aims at describing the self-organization of individuals in a population. For a group of $N$ individuals moving in $\bb R^d$, it reads,
\begin{equation}
\label{CS}
\begin{cases}
\dot {\bs x}_i = \bs v_i\,, \\ \dot {\bs v}_i = \displaystyle{\sum_{j=1}^N \psi(|\bs x_i - \bs x_j|)(\bs v_j - \bs v_i)}\,,
\end{cases} \quad i=1,\ldots,N\,,
\end{equation}
where the pair $(\bs x_i,\bs v_i)\in \bb R^d\times \bb R^d$ denotes the position and velocity of the individual (sometimes named ``particle") $i\in \{1,\ldots,N\}$. Here, $\psi$ is the mutual communication rate between the individuals, that in the original model is chosen of the form,
\begin{equation}
\label{cr}
\psi(u) = \frac {K_0}{(1+u)^\beta}\,, \quad K_0,\beta>0\,.
\end{equation}

The system \eqref{CS} was conceived as a model for \textit{flocking}. Introducing the position and velocity of the center of mass,
\begin{equation}
\label{cm}
\bs x_c := \frac 1N \sum_{j=1}^N \bs x_j\,, \quad \bs v_c := \frac 1N \sum_{j=1}^N \bs v_j\,, 
\end{equation} 
the system exhibits time-asymptotic flocking if
\begin{equation}
\label{flock}
\lim_{t\to+\infty} |\bs v_i(t)-\bs v_c(t)|=0\,, \quad \sup_{t\ge 0} |\bs x_i(t) - \bs x_c(t)| < + \infty \qquad \forall\,i=1,\ldots,N\,.
\end{equation}

It is known that flocking occurs for all initial conditions provided $\beta\le \frac 12$ and for some initial conditions otherwise \cite{CS1,CS2,HL}. In fact, when flocking occurs, also the relative positions converge to an asymptotic configuration: $\bs x_i(t) - \bs x_c(t) \to \bar{\bs x}_i$ as $t\to+\infty$ for suitable $\bar{\bs x}_i$, $i=1,\ldots,N$.

In many real systems, the size of the population could be extremely large, so that a natural issue concerns the behavior of the dynamics when the number of individuals is huge, i.e., in the limit $N\to\infty$. In this direction, a Vlasov-type kinetic model with flocking dissipation has been derived from the many particle (CS)-system \eqref{CS} in a mean-field regime, which is realized by choosing $K_0=\lambda/N$ in \eqref{cr}, for some fixed $\lambda>0$, and letting $N\to\infty$. Furthermore, the resulting kinetic model exhibits time-asymptotic flocking behavior for arbitrary compactly supported initial data.

Beyond the mean-field approximation, a more realistic assumption is that the strength $K_0$ of the communication rate does not depend on the size $N$ of the system. Moreover, it is reasonable to assume that the communication rates have finite range (i.e., replacing $\psi$ in \eqref{cr} by a function with compact support). Indeed, this is consistent with a sharp distinction between the macroscopic size of the system and the typical interaction length among the individuals. In the framework of mean-field approximation, we mention the recent paper \cite{MPT}, where the large time behavior of a continuum alignment dynamics based on (CS)-type interactions with short-range kernel is studied. Finally, to avoid collisions of individuals, it is appropriate to add a repelling force which acts whenever a pair of particles get close, with a strength that increases with this closeness. 

Since the total mass of the system diverges, the existence and locality of the dynamics for infinitely many individuals (in short ``infinite dynamics") is more subtle with respect to the mean-field approximation. 

The existence of the time evolution for systems composed by infinitely many particles moving according to Newton's laws of motion is a classical issue in non-equilibrium statistical mechanics, and several studies have been devoted to this subject, see, e.g., \cite{BPY99,BCM2d,BCMhc,CMP,CC,CMS,DF,FrD77,Lan68,Lan69,P}. For a summary, see for instance the  Appendix 1 of \cite{BCM}. The difficulty of the problem is related to the dimension of the space and the kind of mutual interaction among the particles. In dimension one the problem is solved for almost any kind of interaction, while in dimension two the interaction is required either to be bounded or to diverge at least as an inverse power of the distance between particles. In dimension three, the existence of the infinite dynamics is proved solely in the case of bounded interactions. 

The choice of the initial conditions is a crucial point in the construction of the infinite dynamics. In the framework of non-equilibrium statistical mechanics, for a particle model to be meaningful, the initial conditions have to be chosen in a set which is typical for any reasonable thermodynamic (equilibrium or non-equilibrium) state. Under mild hypothesis on the forces, to this purpose it is sufficient to consider locally finite configurations which have local energy and number of particle fluctuations only of logarithmic order.

When alignment is present, contrary to the case of fundamental interactions, the mutual forces depend also on the velocities and not only on the positions of the particles. This requires a nontrivial adaptation of the techniques developed in the case of classical particle systems. In fact, although its dissipative nature (it makes decreasing the total mechanical energy), the alignment can increase considerably the local energy, whose variation in time turns out to be the quantity one has to control for proving the locality of the dynamics.

In this paper, we rigorously show the good position of the infinite (CS)-type dynamics with repulsions, in the case the individuals are posed in an infinitely extended tube in $\bb R^3$. For the difficulties described above, we are able to handle only this quasi-one-dimensional case. On the other hand, this geometry is suitable to describe situation of real interest (e.g., the collective motion of fishes crowd in a very long channel or people moving in a very long tunnel). Instead, concerning the kind of interaction, the extension of our analysis to the case of rapidly decreasing rates, i.e., $\psi$ as in \eqref{cr} with $\beta\gg 1$, can be treated with some more technicalities.

Beyond existence and locality, much less is known about the long-time behavior of the infinite dynamics. In the case of classical particle systems, we just mention some nontrivial results obtained in recent years, aimed at a microscopic justification of viscous friction \cite{BCM02,BCM03,BCM04,BCM,CM}. In the present context, we are not able to give an example of flocking for the individuals contained in the tunnel. However, at a heuristic level, it seems not impossible. The first step should be to consider a simplified model, in which only the external force confining the individuals in the tube and the mutual alignment are present. Moreover, the latter is assumed with a range larger than the diameter of the tunnel. Then, we choose initial data such that the interaction is essentially binary and we observe that during this binary scattering the energy decreases and the velocities of the two individuals get closer. Of course, the hope is that these effects are large enough for flocking to occur.

The plan of the paper is the following. The next section is devoted to notation, preliminary material, and statement of the result. Section \ref{sec:3} is devoted to the proofs. 

\section{Notation and statement of the results}
\label{sec:2}

The individuals are confined to move freely in an infinitely extended tube of $\bb R^3$. More precisely, we fix a unit vector $\bs n$, a positive real $L>0$, and we  denote by $\bs x^\perp = \bs x - (\bs x\cdot\bs n)\bs n$ the orthogonal projection of $\bs x\in\bb R^3$ along $\bs n$. Then, we consider the tube $\Omega := \{\bs x \in \bb R^3 \colon |\bs x^\perp| < L\}$ of symmetry axis $\bs n$ and radius $L$. Denoting by $(\bs x_i,\bs v_i)\in \Omega\times \bb R^3$ the position and velocity of the $i$-th individual, the phase space $\Gamma$ of the whole system is the collection of sequences $\bs X  =\{(\bs x_i,\bs v_i)\}_{i\in \bb N}$ which are locally finite (i.e., the number of particles inside any bounded region is finite), equipped with the topology of local convergence.

We force the individuals to be confined inside the tube $\Omega$, by requiring that all of them are subjected to a one-body potential of the form
\begin{equation}
\label{V}
\Theta(\bs x) = \frac{\theta_h(|\bs x^\perp|)}{(L-|\bs x^\perp|)^\gamma}\;, \qquad \bs x\in\Omega\;,
\end{equation}
where $\gamma>0$, $h\in (0,L)$, and $\theta_h(s)$, $s\in \bb R^+$, is a non-negative, twice differentiable function, identically zero for $s\le h$ and strictly positive at $s=L$. 

Our task is to prove the well-posedness of the following system of infinitely many ordinary differential equations,
\begin{equation}
\label{CS1}
\begin{cases}
\dot {\bs x}_i(t) = \bs v_i(t)\,, \\ \dot {\bs v}_i(t) = \sum_j \psi_{i,j}(t)\, (\bs v_j(t) - \bs v_i(t)) + \bs F_i(\bs X(t))\,,
\end{cases}
\quad i\in \bb N\,,
\end{equation}
where we used the sharp notation
\[
\psi_{i,j}(t) = \psi(|\bs x_i(t) - \bs x_j(t)|)
\]
for the communication rates, and the force $\bs F_i$ is given by
\begin{equation}
\label{F}
\bs F_i(\bs X) = - \sum_{j\ne i}\nabla U (\bs x_i - \bs x_j) - \nabla\Theta(\bs x_i)\,.
\end{equation}
Above, the function $\psi \in C(\bb R)$ is assumed symmetric, non-negative, with compact support. The potential $U$ is non-negative, symmetric, short-range, of the form
\begin{equation}
\label{U}
U(\bs x) =U_1(\bs x) + a |\bs x|^{-b}\;,
\end{equation}
where $a\ge 0$, $b>0$, and $U_1$ is twice differentiable and symmetric. In particular, there exists $\bar r>0$ such that the supports of $\psi$ and $U$ are contained in the interval $[-\bar r,\bar r]$. If $U$ is finite at the origin, i.e., if $a=0$, we assume $U(0)>0$, which guarantees $U$ to be superstable \cite{Ruelle}. Under these assumptions, we will show that the system \eqref{CS1} determines a differentiable flow on a non-trivial subset $\mc X$ of $\Gamma$.

In the case of standard potential forces, the conservation of the particle number is sufficient to prove the existence of the infinite dynamics in dimension $d=1$ (or in quasi-one-dimensional regions like the infinitely extended tube $\Omega$), while in dimension $d=2$ the crucial tool is the energy conservation. In the present case, since the force depends on the velocity, the conservation of the particle number is not sufficient to prove the result and we need to use the energy, which is dissipated along the motion in this case.

Given a configuration $\bs X = \{(\bs x_i,\bs v_i)\}_{i\in \bb N}$, for any $\mu\in\bb R$ and $R>0$, we consider the quantity
\begin{equation}
\label{QR1}
Q(\bs X;\mu,R) := \sum_i \chi_i(\mu,R) \,
\bigg\{ \frac{\bs v_i^2}{2} + \frac 12 \sum_{j:j\ne i} U(\bs x_i - \bs x_j) + \Theta(\bs x_i) + 1 \bigg\}\;,
\end{equation}
where $\chi_i(\mu,R) = \chi(|\bs x_i\cdot\bs n - \mu|\le R)$. Clearly, $Q(\bs X;\mu,R)$ is the sum of the energy and number of particles in the finite region 
\[
\Omega(\mu,R) := \{\bs x \in\Omega \colon |\bs x\cdot\bs n - \mu |\le R\}\,,
\]
and we allow initial data with logarithmic divergences in local energy and density. More precisely, letting
\begin{equation}
\label{Q1}
Q(\bs X) := \sup_\mu  \sup_{R:R>\log(\rme+|\mu|)}  \frac{Q(\bs X;\mu,R)}{2R}\;,
\end{equation}
we restrict to the set of configurations
\begin{equation}
\label{mcx}
\mc X := \{\bs X\in \Gamma \colon Q(\bs X)< \infty\}\,.
\end{equation}

The infinite dynamics is constructed as a limit of the so-called $n$-partial dynamics, which is defined in the following way. Given $\bs X=\{(\bs x_i,\bs v_i)\}_{i\in \bb N} \in\mc X$ and $n\in\bb N$  let $I_n := \{i\in\bb N \colon \bs x_i \in \Omega(0,n)\}$. The $n$-partial dynamics $t\mapsto \bs X^{(n)}(t) = \{(\bs x^{(n)}_i(t),\bs v^{(n)}_i(t))\}_{i\in I_n}$ is defined as the solution to the Cauchy problem,
\begin{equation}
\label{CSn}
\begin{cases}
\dot {\bs x}^{(n)}_i(t) = \bs v^{(n)}_i(t)\,, \\ \dot {\bs v}^{(n)}_i(t) = \sum_{j\in I_n} \psi^{(n)}_{i,j}(t)\, (\bs v^{(n)}_j(t)-\bs v^{(n)}_i(t)) + \bs F^{(n)}_i(\bs X^{(n)}(t))\,, \\ X^{(n)}(0) = \{(\bs x_i,\bs v_i)\}_{i\in I_n}\,,
\end{cases}
\end{equation}
where $i\in I_n$ and
\begin{align*}
\psi^{(n)}_{i,j}(t) & = \psi(|\bs x^{(n)}_i(t) - \bs x^{(n)}_j(t)|)\,, \\
\bs F^{(n)}_i(\bs X^{(n)}) & = - \sum_{j\in I_n: j\ne i} \nabla U( (\bs x^{(n)}_j - \bs x^{(n)}_i) - \nabla \Theta(\bs x^{(n)}_i)\,.
\end{align*}

\begin{theorem}
\label{thm:0}
For $\bs X\in\mc X$ the following limits exist,
\begin{equation}
\label{eu1}
\lim_{n\to \infty} \bs x^{(n)}_i(t) = \bs x_i(t)\;, \quad \lim_{n\to \infty} \bs v^{(n)}_i(t) = \bs v_i(t)\;, \quad i\in \bb N\,.
\end{equation}
Moreover, the flow $t\mapsto \bs X(t) = \{(\bs x_i(t),\bs v_i(t))\}_{i\in\bb N}$ is the unique (global) solution to \eqref{CS1} such that $	\bs X(t)\in\mc X$. 
\end{theorem}

We conclude the section with a notation warning: in the sequel, if not further specified, we shall denote by $C$ a generic positive constant whose numerical value may change from line to line and it may possibly depend only on the interactions $\Theta,U$ and on the communication rate $\psi$.

\section{Proofs}
\label{sec:3}

A basic tool in the proof of Theorem \ref{thm:0} is an estimate on the growth in time of the local density and energy, which is the content of the following lemma.

\begin{lemma}
\label{lem:1}
There exists a constant $K>0$ such that, for any $\bs X\in\mc X$ and $n\in\bb N$,
\begin{equation}
\label{K0}
\sup_\mu Q(\bs X^{(n)}(t);\mu,R_n(t)) \le K Q(\bs X) R_n(t) \qquad
\forall\, t\ge 0\;,
\end{equation}
where 
\begin{equation}
\label{Rn1}
R_n(t) := (1+\bar r) \log(\rme+n) + \int_0^t\!\rmd s\; M_n(s)
\end{equation}
and
\begin{equation}
\label{Vn1}
M_n(t) := 1+ V_n(t)^2\,, \quad V_n(t) := \max_{i\in I_n} \sup_{s\in [0,t]} |\bs v^{(n)}_i(s)|\,.
\end{equation}
\end{lemma}

\begin{remark}
\label{rem:Rn}
An estimate like \eqref{K0} is a key ingredient for proving existence and locality of the time evolution of infinitely many interacting particles. It is worthwhile to notice that in the present context, with respect to the case of standard inertial particles, the quantity $R_n(t)$ is not simply given by the maximal displacement of the particles (this would be the case choosing $M_n(t) = V_n(t)$) but quite larger. Indeed, since here the forces depend also on the velocities, this choice is mandatory to recover the extensive property \eqref{K0} of the local energy.
\end{remark}

\begin{proof}[Proof of Lemma \ref{lem:1}] We introduce a mollified version of $Q(\bs X;\mu,R)$,
\begin{equation}
\label{WR}
W(\bs X;\mu,R) := \sum_i f_i^{\mu,R} \, \left\{ \frac{\bs v_i^2}{2} + \frac 12 \sum_{j:j\ne i} U(\bs x_i - \bs x_j) + \Theta(\bs x_i) +1 \right\},
\end{equation}
where 
\begin{equation}
\label{fi}
f^{\mu,R}_i = f\left(\frac{|\bs x_i\cdot \bs n - \mu|}{R}\right)
\end{equation}
and $f\in C^\infty(\bb R_+)$ is not increasing and satisfies: $f(r)=1$ for $r\in [0,1]$, $f(r)=0$ for $r\ge 2$, and $|f'(r)| \le 2$. Clearly,
\begin{equation}
\label{a4}
Q(\bs X;\mu,R) \le W(\bs X;\mu,R) \le Q(\bs X;\mu,2R)\,.
\end{equation}
For $0\le s \le t$, we define
\begin{equation}
\label{a5}
R_n(t,s) :=(1+\bar r) \log(\rme+n) + \int_0^t\!\rmd\tau\; M_n(\tau) + \int_s^t\!\rmd\tau\; M_n(\tau)
\end{equation}
(note that $R_n(t,t) = R_n(t)$ and $R_n(t,0) \le 2R_n(t)$) and compute
\begin{equation}
\label{a6}
\partial_s W(\bs X^{(n)}(s); \mu, R_n(t,s)) = \sum_i \left[\kappa_i(t,s) \varepsilon_i(s) + f^{\mu,R_n(t,s)}_i \dot \varepsilon_i(s)\right]\,,
\end{equation}
where, denoting by $e_i^\mu(s)$ the sign of $\bs x_i(s)\cdot \bs n - \mu$,
\begin{align*} 
\kappa_i(t,s) & = f' \left( \frac{|\bs x_i(s)\cdot \bs n- \mu|}{R_n(t,s)} \right) \left[ \frac{e_i^\mu(s) \bs v_i(s)\cdot \bs n}{R_n(t,s)} - \frac{\partial_s R_n(t,s)}{R_n(t,s)^2}|\bs x_i(s)\cdot \bs n - \mu|\right], \\ \varepsilon_i (s) & = \frac{\bs v_i(s)^2}{2} +\frac 12 \sum_{j\in I_n:j \ne i} U(\bs x_i-\bs x_j) + \Theta(\bs x_i) +1\,,
\end{align*}
and, to simplify notation, we have omitted the explicit dependence on $n$ of $\bs x_i$, $\bs v_i$, $\kappa_i$, and $\varepsilon_i$.

We observe that $f'(|y|)\le 0$, $f'(|y|) = 0$ if $|y|\le 1$, $\partial_s R_n(t,s) = - M_n(s)$, and $|\bs v_i(s)\cdot \bs n| \le V_n(s) \le M_n(s)$, so that $\kappa_i(t,s) \le 0$. On the other hand, from the equations of motion,
\begin{align*}
\dot \varepsilon_i(s) & = \sum_{j\in I_n} \psi_{i,j}(s) (\bs v_j(s) - \bs v_i(s)) \cdot \bs v_i(s) \\ & \quad - \sum_{j\in I_n:j\ne i} \nabla U(\bs x_i(s)-\bs x_j(s)) \cdot \frac{\bs v_i(s) + \bs v_j(s)}{2}\,.
\end{align*}
Then, by \eqref{a6} and using that $\psi_{i,j}=\psi_{j,i}$ and $\nabla U$ is odd,
\[
\begin{split}
& \partial_s W(\bs X^{(n)}(s); \mu, R_n(t,s)) \\ & \quad \le \sum_{i,j\in I_n:i\ne j} \Big(f^{\mu,R_n(t,s)}_i - f^{\mu,R_n(t,s)}_j\Big) \nabla U (\bs x_i(s)-\bs x_j(s))\cdot \frac{\bs v_i(s)}{2}\\ &  \quad\quad + \frac 12 \sum_{i,j\in I_n} \psi_{i,j}(s) (\bs v_j(s) - \bs v_i(s)) \cdot \Big[\bs v_i(s) f^{\mu,R_n(t,s)}_i - \bs v_j(s) f^{\mu,R_n(t,s)}_j\Big] \\ & \quad = \sum_{i,j\in I_n:i\ne j} \Big(f^{\mu,R_n(t,s)}_i - f^{\mu,R_n(t,s)}_j\Big)  \nabla U (\bs x_i(s)-\bs x_j(s)) \cdot \frac{\bs v_i(s)}{2}\\ & \quad\quad - \sum_{i,j\in I_n}  f^{\mu,R_n(t,s)}_i \psi_{i,j}(s) (\bs v_j(s) - \bs v_i(s))^2 \\ & \quad \quad + \frac 12 \sum_{i,j\in I_n} \Big(f^{\mu,R_n(t,s)}_i - f^{\mu,R_n(t,s)}_j\Big) \psi_{i,j}(s) (\bs v_j(s) - \bs v_i(s)) \cdot \bs v_j(s) \\ & \quad \le J_1 + J_2\,,
\end{split}
\]
with 
\[
\begin{split}
J_1 & =  \sum_{i,j\in I_n:i\ne j} \Big(f^{\mu,R_n(t,s)}_i - f^{\mu,R_n(t,s)}_j\Big) \nabla U (\bs x_i(s)-\bs x_j(s))\cdot \frac{\bs v_i(s)}{2}\,,\\ J_2 & = \frac 12 \sum_{i,j\in I_n}\Big(f^{\mu,R_n(t,s)}_i - f^{\mu,R_n(t,s)}_j\Big)  \psi_{i,j}(s) (\bs v_j(s) - \bs v_i(s)) \cdot \bs v_j(s) \,.
\end{split}
\]
From \eqref{U} we have $ |\bs x|\, |\nabla U(\bs x)| \le C[1+U(\bs x)]$ for any $\bs x \ne 0$. Then, by the inequalities
\[
\big| f^{\mu,R}_i - f^{\mu,R}_j \big| \le 2 \,
\frac{|\bs x_i - \bs x_j|}{R} \big[ \chi_i(\mu,2R) + \chi_j(\mu,2R)\big]\,,
\]
$|\bs v_i(s)| \le M_n(s) = - \partial_s R_n(t,s)$, and $R_n(t,s) > \bar r$, we have,
\begin{equation}
\label{a7}
\begin{split}
J_1 & \le - C \frac{\partial_s R_n(t,s)}{R_n(t,s)} \sum_{i\ne j} [1+U(\bs x_i(s)-\bs x_j(s))]  \\ & \qquad \times \chi_i(\mu,4R_n(t,s)) \chi_j(\mu,4R_n(t,s)) \chi_{i,j}(s)\;,
\end{split}
\end{equation}
where we shortened $\chi_{i,j}(s) = \chi (|\bs x_i(s)-\bs x_j(s)|\le \bar r)$. Analogously, since $|(\bs v_j(s) - \bs v_i(s)) \cdot \bs v_j(s)| \le C V_n(s)^2 \le C M_n(s) = - C \partial_s R_n(t,s)$, we also have,
\begin{equation}
\label{a8} 
J_2 \le - C \frac{\partial_s R_n(t,s)}{R_n(t,s)} \sum_{i,j\in I_n} \chi_i(s,4R_n(t,s)) \chi_j(s,4R_n(t,s)) \chi_{i,j}(s)\,.
\end{equation}
As $U$ is a superstable potential, by arguing as in the proof of \cite[ Eq.~(2.15)]{CMP}, the double sums in the right hand side of \eqref{a7} and \eqref{a8} can be bounded by $C W(\bs X^{(n)}(s); \mu, 4 R_n(t,s))$; moreover, setting
\begin{equation}
\label{a9}
W(\bs X;R) := \sup_\mu W(\bs X;\mu,R)\;,
\end{equation}
it can be proved that 
\begin{equation}
\label{a10}
W(\bs X;\mu,2R) \le CW(\bs X;R)
\end{equation}
(see e.g.\ \cite{CM,CMP}). In conclusion,
\begin{equation*}
\partial_s W(\bs X^{(n)}(s); \mu, R_n(t,s)) \le - C \frac{\partial_s R_n(t,s)}{R_n(t,s)} W(\bs X^{(n)}(s); R_n(t,s))\;,
\end{equation*}
from which, by integrating and taking the supremum on $\mu$, 
\[
\begin{split}
& W(\bs X^{(n)}(s); R_n(t,s)) \le W(\bs X^{(n)}(0); R_n(t,0)) \\ & - \qquad C \int_0^s\!\rmd\tau\; \frac{\partial_\tau R_n(t,\tau)}{R_n(t,\tau)} W(\bs X^{(n)}(\tau); R_n(t,\tau))\,,
\end{split}
\]
whence 
\begin{equation*}
W(\bs X^{(n)}(s);R_n(t,s)) \le W(\bs X^{(n)}(0);R_n(t,0)) \left( \frac{R_n(t,0)}{R_n(t,s)}\right)^C.
\end{equation*}
Setting $s=t$ and using that $R_n(t,0)\le 2 R_n(t,t) = 2R_n(t)$,  
\begin{equation*}
W(\bs X^{(n)}(t);R_n(t)) \le C \,W(\bs X^{(n)}(0);R_n(t))\;.
\end{equation*}
Then, from \eqref{a4}, \eqref{a9}, and definition \eqref{Q1}, we conclude that
\begin{align*}
\sup_\mu Q(\bs X^{(n)}(t);\mu,R_n(t)) & \le  C \, W (\bs X^{(n)}(0);R_n(t)) 
\\ & \le  C \sup_\mu Q(\bs X^{(n)}(0);\mu,2R_n(t)) \\ & \le  4 C  Q(\bs X) R_n(t)\;,
\end{align*}
which proves \eqref{K0}.
\end{proof}

\begin{corollary}
\label{cor:1}
For each $\bs X\in\mc X$, $n\in\bb N$, and $t\ge 0$ there exists a function $H_t=H_t(Q(\bs X))$ such that 
\begin{equation}
\label{stima Vn}
V_n(t) \le H_t \sqrt{\log(\rme + n)}\,.
\end{equation}
\end{corollary}

\begin{proof}
From \eqref{QR1}, \eqref{Vn1}, and \eqref{K0} we have 
\[
M_n(t) \le 2\sup_\mu Q(\bs X^{(n)}(t);\mu,R_n(t)) \le 2KQ(\bs X)R_n(t)\,.
\]
Therefore, from the definition \eqref{Rn1},
\[
R_n(t) \le (1+\bar r) \log(\rme+n) + 2KQ(\bs X) \int_0^t\!\rmd s\; R_n(s)\,,
\]
which implies, by Grönwall's inequality, 
\begin{equation}
\label{Rn<}
R_n(t) \le  (1+\bar r) \log(\rme+n) \exp\Big(2KQ(\bs X) t\big)\,,
\end{equation}
so that
\[
M_n(t) \le (1+\bar r) \log(\rme+n) 2KQ(\bs X) \exp\Big(2KQ(\bs X) t\big)\,.
\]
As $V_n(t) \le \sqrt{M_n(t)}$, the inequality \eqref{stima Vn} is thus proved with
\[
H_t = \sqrt{2(1+\bar r) KQ(\bs X)} \exp\Big(KQ(\bs X)t\Big)\,.
\]
\end{proof}

\begin{proof}[Proof of Theorem \ref{thm:0}] 
Let
\begin{equation}
\label{de}
\delta_i(n,t) := |\bs x^{(n)}_i(t) - \bs x^{(n-1)}_i(t)| + |\bs v^{(n)}_i(t) - \bs v^{(n-1)}_i(t)|\,.
\end{equation}
From the equations of motion in integral form it follows that, for any $i\in I_{n-1}$, 
\begin{align}
\label{de<}
|\delta_i(n,t)|&  \le \int_0^t\! \rmd s \int_0^s\! \rmd s'\, |\dot{\bs v}^{(n)}_i(s') - \dot{\bs v}^{(n-1)}_i(s')| + \int_0^t\! \rmd s\,  |\dot{\bs v}^{(n)}_i(s) - \dot{\bs v}^{(n-1)}_i(s)| \nonumber \\ & \le (1+t)  \int_0^t\! \rmd s\,  |\dot{\bs v}^{(n)}_i(s) - \dot{\bs v}^{(n-1)}_i(s)| \nonumber \\ & \le (1+t) \int_0^t\!\rmd s\; G^{(n)}_i(s)\,,
\end{align}
where
\[
\begin{split}
& G^{(n)}_i(s) := \big| \nabla \Theta(\bs x^{(n)}_i(s)) - \nabla \Theta(\bs x^{(n-1)}_i(s))\big| \\ & \quad + \sum_{j\in I_{n-1}:\,j\ne i} \Big|\nabla U(\bs x^{(n)}_i(s)-\bs x^{(n)}_j(s)) - \nabla U(\bs x^{(n-1)}_i(s)-\bs x^{(n-1)}_j(s))\Big| \\ & \quad +\sum_{j\in I_{n-1}}\Big| \psi^{(n)}_{i,j}(s)\, (\bs v^{(n)}_j(s) - \bs v^{(n)}_i(s)) - \psi^{(n-1)}_{i,j}(s)\, (\bs v^{(n-1)}_j(s) - \bs v^{(n-1)}_i(s))\Big| \\ & \quad + \sum_{j\in I_n\setminus I_{n-1}} \Big| \nabla U(\bs x^{(n)}_i(s)-\bs x^{(n)}_j(s))\Big| +  \Big|\psi^{(n)}_{i,j}(s)\, (\bs v^{(n)}_j(s)-\bs v^{(n)}_i(s)) \Big|\,.
\end{split}
\]
By \eqref{stima Vn}, each particle $i\in I_n$ may interact during the time $[0,t]$ only with the particles $j$ such that $|\bs x_j-\bs x_i|\le p_n(t)$, with
\begin{equation}
\label{pn1}
p_n(t) := \bar r + 2t H_t \sqrt{\log(\rme+n)}\,.
\end{equation}
We now fix $k\in\bb N$ and define
\begin{equation}
\label{nk}
n(k) := \min\{m\in\bb N \colon n > \bar r + k + p_n(t) \quad \forall\, n\ge m \}\;.
\end{equation}
For $n\ge n(k)$ each particle $i\in I_k$ does not interact, during the time $[0,t]$, with the particles $j\in I_n \setminus I_{n-1}$. Otherwise stated, for any $i\in I_k$ and $s\in [0,t]$ the last sum in the definition of $G^{(n)}_i(s)$ is equal to zero. To control the first terms we observe that, we since $\Theta$ and $U$ are of the form \eqref{V} and \eqref{U}, respectively, we have, setting $\zeta := (\gamma + 2)/\gamma$ and $\eta :=(b+2)/b$, 
\begin{align*}
|\nabla \Theta(\bs x) - \nabla \Theta(\bs y)| & \le C \big[\Theta(\bs x)^\zeta + \Theta(\bs y)^\zeta + \chi_1(\bs x) + \chi_1(\bs y) \big] |\bs x - \bs y|\,, \\
|\nabla U(\bs x) - \nabla U(\bs y)| & \le C\left[U(\bs x)^\eta + U(\bs y)^\eta + \chi(|\bs x|\le \bar r) + \chi(|\bs y|\le \bar r) \right] |\bs x - \bs y |\;,
\end{align*}
where  we used the shorten notation $\chi_1(\bs z) := \chi(|\bs z^\perp|\ge h)$. To control the second one, we instead use that
\[
\begin{split}
& \Big| \psi^{(n)}_{i,j}(\bs v^{(n)}_j - \bs v^{(n)}_i) - \psi^{(n-1)}_{i,j}\, (\bs v^{(n-1)}_j - \bs v^{(n-1)}_i)\Big| \le 2 V_n \big|\psi^{(n)}_{i,j} - \psi^{(n-1)}_{i,j} \big| \\ & \qquad \qquad  + \psi^{(n-1)}_{i,j} \big| \bs v^{(n)}_j - \bs v^{(n-1)}_j - \bs v^{(n)}_i + \bs v^{(n-1)}_i\big|\,.
\end{split}
\]
Therefore, recalling the definition \eqref{de} and using that $\psi$ is a smooth function with compact support, we obtain, for any $n\ge n(k)$, $s\ge 0$, and $i\in I_k$,
\begin{equation}
\label{Gni}
\begin{split}
& G^{(n)}_i(s) \le C \big[\Theta(\bs x^{(n)}_i(s))^\zeta + \Theta(\bs x^{(n-1)}_i(s))^\zeta + 1 \big]\delta_i(n,s) \\ & \qquad + C {\sum_j}^* (1+V_n(s)) [\delta_i(n,s) + \delta_j(n,s)] \\ & \qquad + C{\sum_{j:j\ne i}}^* \left[U(\bs x^{(n)}_i(s) - \bs x^{(n)}_j(s))^\eta + U(\bs x^{(n-1)}_i(s)-\bs x^{(n-1)}_j(s))^\eta + 1 \right] \\ & \qquad\qquad \times \left[\delta_i(n,s) + \delta_j(n,s)\right]\;,
\end{split}
\end{equation}
where $\sum^* $ denotes the sums restricted to all the particles $j\in I_{n-1}$ closer than $\bar r$ to $\bs x^{(n)}_i(s)$ or $\bs x^{(n-1)}_i(s)$. 
Therefore, introducing
\begin{equation}
\label{uk1}
u_k(n,t) := \sup_{i\in I_k} \delta_i(n,t),
\end{equation}
by \eqref{de<} and recalling the definition \eqref{QR1} we have, for any $t\ge 0$,
\[
\begin{split}
u_k(n,t) & \le C\int_0^t\!\rmd s\; \sup_\mu\bigg\{\left[Q(\bs X^{(n)}(s);\mu,\bar r) + Q(\bs X^{(n-1)}(s);\mu,\bar r)\right] (1+V_n(s)) \\ & \quad +
\left[Q(\bs X^{(n)}(s);\mu,\bar r) + Q(\bs X^{(n-1)}(s);\mu,\bar r)\right]^{\eta_1}\bigg\} u_{k_1}(n,s)\,,
\end{split}
\]
where $\eta_1=\max\{\zeta;\eta\}$ and $k_1 =\lfloor k+p_n(t)  \rfloor + 1$. Now, since $R_n(t) \ge R_{n-1}(t) \ge \bar r$, using Lemma \ref{lem:1} and Corollary \ref{cor:1}, we get
\begin{equation}
\label{uk<}
u_k(n,t) \le g_n(t) \int_0^t\!\rmd s\; u_{k_1}(n,s)\,,
\end{equation}
with
\begin{equation}
\label{gn}
g_n(t) = A_t \log^{\eta'}(\rme+n)\,,
\end{equation}
$\eta' = \max\{\eta_1;3/2\}$, and $A_t = A_t(Q(\bs X))>0$ large enough. 

Setting $k_q = \lfloor k_{q-1}+p_n(t) \rfloor + 1$, $q\in\bb N$, and $k_0=k$, we can iterate the inequality \eqref{uk<} $\ell$ times, with 
\begin{equation}
\label{ell}
\ell := \left\lfloor \frac{n-k-1}{1+p_n(t)} \right\rfloor
\end{equation}
(which ensures $n>n(k_{\ell-1})$). Since
\begin{equation}
\label{an}
u_k(n,t) \le a_n(t) := 2(1+t) H_t \sqrt{\log(\rme+n)}\,,
\end{equation}
we finally get,
\begin{equation}
\label{uk<fi}
u_k(n,t) \le a_n(t) \frac{[g_n(t)t]^\ell}{\ell !}\;.
\end{equation}
Recalling the definitions \eqref{de} and \eqref{uk1}, the existence of the infinite dynamics via the limits \eqref{eu1} now follows from the absolute convergence, uniform on compact time intervals, of the series $\sum_n u_k(n,t)$, which is a straightforward consequence of \eqref{uk<fi}. The proof of uniqueness can be done in a very similar way and it is therefore omitted.

It remains to show that $\bs X(t)\in \mc X$ for any $t>0$, i.e., that $Q(\bs X(t)) < +\infty$, see \eqref{Q1}. By \eqref{a4} it is enough to estimate $W(\bs X(t);\mu,R)$ instead of $Q(\bs X(t);\mu,R)$. In what follows, we fix $t > 0$, and $R,\mu$ with $R>\log(\rme+|\mu|)$. Given $\alpha \ge 1$ to be fixed large enough later, let $n_0 = \left\lfloor \alpha \, \rme^{2R} \right\rfloor+1$ (where $\left\lfloor x \right\rfloor$ denotes the integer part of the real $x$). Since $\log(\rme+n_0)>R$, by \eqref{Rn1} we have $R_{n_0}(t)>R$ and therefore, by \eqref{K0} and \eqref{a4},
\[
\begin{split}
W(\bs X^{(n_0)}(t);\mu,R) & \le W(\bs X^{(n_0)}(t);\mu,R_{n_0}(t)) \le Q(\bs X^{(n_0)}(t) ;\mu, 2R_{n_0}(t))  \\ & \le 2K \rme^{Kt}Q(\bs X) R_{n_0}(t) \\ & \le  2K \rme^{Kt}Q(\bs X) \left[1+\bar r + t (1+ H_t^2)) \right] \log(\rme+n_0)\,,
\end{split}
\]
where in the last inequality we used again \eqref{Rn1}, together with \eqref{Vn1} and \eqref{stima Vn}. From the choice of $n_0$ we conclude that there is $B_t = B_t(Q(\bs X))$ such that
\begin{equation}
\label{Wn0}
W(\bs X^{(n_0)}(t);\mu,R) \le B_t (R+\log\alpha)\,.
\end{equation}
On the other hand,
\begin{equation}
\label{WWn}
\begin{split}
W(\bs X(t);\mu,R) & \le W(\bs X^{(n_0)}(t);\mu,R) \\ & + \sum_{n>n_0} |W(\bs X^{(n)}(t);\mu,R) - W(\bs X^{(n-1)}(t);\mu,R)|\,.
\end{split}
\end{equation}
Let us estimate the sum on the right-hand side of \eqref{WWn}. We have,
\begin{eqnarray}
\label{ka0}
&& |W(\bs X^{(n)}(t);\mu,R) - W(\bs X^{(n-1)}(t);\mu,R)| \nonumber \\&& \quad \le \, \sum_i f\left(\frac{|\bs x^{(n)}_i(t)\cdot \bs n - \mu|}{R}\right) \big|\varepsilon^{(n)}_i
- \varepsilon^{(n-1)}_i \big| \nonumber\\&&\quad + \sum_i \left| f\left(\frac{|\bs x^{(n)}_i(t)\cdot \bs n - \mu|}{R}\right) - f\left(\frac{|\bs x^{(n-1)}_i(t)\cdot \bs n - \mu|}{R}\right) \right| \varepsilon^{(n-1)}_i,
\qquad\quad
\end{eqnarray}  
where
\begin{equation*}
\varepsilon^{(n)}_i = \frac{|\bs v^{(n)}_i(t)|^2}2 + \frac 12 \sum_{j\in I_n : j\ne i} U(\bs x^{(n)}_i(t) - \bs x^{(n)}_j(t)) + \Theta(\bs x_i) + 1\;.
\end{equation*}
If $|\bs x^{(n)}_i(t)-\mu|\le 2R$ then all the particles $j\in I_n$ such that $|\bs x^{(n)}_i(t)-\bs x^{(n)}_j(t)| \le \bar r$ or $|\bs x^{(n-1)}_i(t)-\bs x^{(n-1)}_j(t)| \le \bar r$ are initially contained in the intersection of the tube $\Omega$ with the ball of center $\mu$ and radius $\Gamma_n(t)$, with $\Gamma_n(t) = C [ R + p_n(t)]$, see \eqref{pn1}. In particular, by choosing $\alpha$ large enough, for any $n \ge n_0$ each particle $i$ such that $|\bs x^{(n)}_i(t) - \mu| \le 2R$ does not interact with the particles $j\in I_n \setminus I_{n-1}$, so that 
\[
\begin{split}
& \big|\varepsilon^{(n)}_i - \varepsilon^{(n-1)}_i \big| \le C \frac{|\bs v^{(n)}_i(t)|+|\bs v^{(n-1)}_i(t)|}2 \delta_i(n,t) \\ & \qquad +C \big[\Theta(\bs x^{(n)}_i(s))^\zeta + \Theta(\bs x^{(n-1)}_i(s))^\zeta + 1 \big]\delta_i(n,s)\\ & \qquad + C{\sum_{j:j\ne i}}^* \left[U(\bs x^{(n)}_i(s) - \bs x^{(n)}_j(s))^\eta + U(\bs x^{(n-1)}_i(s)-\bs x^{(n-1)}_j(s))^\eta + 1 \right] \\ & \qquad\qquad \times [\delta_i(n,s) + \delta_j(n,s)]\;,
\end{split}
\]
where the notation $\sum_{j:j\ne i}^*$ is as in \eqref{Gni}. In particular, the particles involved in this sum are initially contained in the intersection of the tube $\Omega$ with the ball of center $\mu$ and radius $\Gamma_n(t)$. Therefore, by arguing as in obtaining \eqref{uk<} we conclude that, setting $\Delta_n(t) := \max\{\delta_i(n,t) : |\bs x_i -\mu | \le \Gamma_n(t)\}$, there exists $C_t = C_t(Q(\bs X))$ such that if $|\bs x_i-\mu| \le \Gamma_n(t)$ then, for any $n>n_0$,
\begin{equation}
\label{ka1}
\big|\varepsilon^{(n)}_i - \varepsilon^{(n-1)}_i \big| \le C_t \log^{\eta_1} (\rme+n) \Delta_n(t)\,,
\end{equation}
with $\eta_1 = \max\{\eta;\zeta\}$ as before. On the other hand,
\begin{eqnarray}
\label{ka2}
&& \left| f\left(\frac{|\bs x^{(n)}_i(t) - \mu|}{R}\right) - 
f\left(\frac{|\bs x^{(n-1)}_i(t) - \mu|}{R}\right) \right| 
\nonumber\\&& \qquad\qquad 
\le 2 \frac{|\bs x^{(n)}_i(t)-\bs x^{(n-1)}_i(t)|}{R} 
\chi\big(|\bs x^{(n-1)}_i(t) - \mu| \le \delta_i(n,t) + 2R \big)
\nonumber\\&& \qquad\qquad \le C \chi\big(|\bs x^{(n-1)}_i(t) - \mu|
\le \Delta_n(t) +2R\big) \, \Delta_n(t)\,. 
\end{eqnarray}
To estimate $\Delta_n(t)$ we can use \eqref{uk<fi} with $k = |\mu| + \Gamma_n(t)$. Recalling \eqref{ell} and that $n_0 = \left\lfloor \alpha \, \rme^{2R} \right\rfloor+1$, we choose $\alpha$ large enough so that $|\mu|+R \le \frac 12 \log(\rme + n)$ and $n-k-1 > \frac 12 n$ for any $n>n_0$. Under this assumptions, we can find $D_t=D_t(Q(\bs X))>0$ such that 
\begin{equation}
\label{ka<}
\Delta_n(t) \le D_t \exp\left[-\frac{n}{D_t \sqrt{\log(\rme+n)}}\right] \qquad \forall\, n>n_0 
\end{equation}
(in particular $\Delta_n(t) \le C$). Then, inserting the above bounds in \eqref{ka0}, 
\[
\begin{split}
&|W(\bs X^{(n)}(t);\mu,R) - W(\bs X^{(n-1)}(t);\mu,R)| \\ & \qquad \qquad  \le C_t \log^{\eta_1}(\rme+n) W(\bs X^{(n)}(t);\mu,2R)  \Delta_n(t) \\ & \qquad \qquad  \quad  +  C W(\bs X^{(n-1)}(t);\mu, \Delta_n(t) + 2R)  \Delta_n(t)\,.
\end{split}
\]
As $n>n_0$ and $\log(\rme+n_0)>R$, from \eqref{Rn1} we have $R_n(t)>R$ and therefore, by \eqref{K0} and \eqref{a4},
\[
\begin{split}
 &|W(\bs X^{(n)}(t);\mu,R) - W(\bs X^{(n-1)}(t);\mu,R)|  \\ & \qquad \qquad  \le C_t \log^{\eta_1}(\rme+n) W(\bs X^{(n)}(t);\mu,2R_n(t))  \Delta_n(t) \\ & \qquad \qquad  \quad  + C W(\bs X^{(n-1)}(t);\mu, 2R_{n-1}(t) + C)  \Delta_n(t) \\ & \qquad \qquad  \le E_t \big[ \log^{\eta_1}(\rme+n) R_n(t) + R_{n-1}(t) \big] \Delta_n(t)\,, 
\end{split}
\]
where $E_t=E_t(Q(\bs X))$. In view of Lemma \ref{lem:1} and Corollary \ref{cor:1}, from \eqref{ka<} we deduce that the sum in the right-hand side of \eqref{WWn} is bounded by a constant, independent of $\mu\in \bb R$ and $R>0$ provided $R>\log(\rme+|\mu|)$. Therefore, since $Q(\bs X(t);\mu,R) \le W(\bs X(t);\mu,R)$, from \eqref{Wn0} we conclude that $\bs X(t) \in \mc X$. 
\end{proof}

\section*{Conflict of Interest} The authors declare that they have no conflict of interest.

\end{document}